\documentclass[12pt]{article}

\usepackage{amsmath}
\usepackage{graphicx,psfrag,epsf}
\usepackage{enumerate}
\usepackage{psfrag,epsf,tikz}
\usepackage{url} 
\usepackage{amsmath,amssymb,amsfonts,amsthm,mathtools,mathrsfs}

\usepackage{algorithm}
\usepackage{algpseudocode}

\usepackage{bm,bbm}
\usepackage{color}
\usepackage{subfigure}
\usepackage{algorithm,algpseudocode}
\usepackage[symbol]{footmisc}
\usepackage{scalefnt}
\usepackage{authblk} 
\usepackage{multirow,centernot}
\usepackage[colorlinks=true, citecolor=blue, urlcolor=blue]{hyperref}
\usepackage{natbib}
\usepackage{dsfont}
\usepackage[title]{appendix}
\input cyracc.def

\usepackage[left=1in,top=1.1in,right=0.5in,bottom=1in]{geometry}

\theoremstyle{definition}

\newtheorem{theorem}{Theorem}[section]

\newtheorem{corollary}[theorem]{Corollary}
\newtheorem{example}{Example}[section]
\newtheorem{definition}{Definition}[section]
\newtheorem{lemma}[theorem]{Lemma}
\newtheorem{proposition}[theorem]{Proposition}
\newtheorem{remark}[theorem]{Remark}

\makeatletter
\def\@seccntformat#1{\@ifundefined{#1@cntformat}%
	{\csname the#1\endcsname\quad}
	{\csname #1@cntformat\endcsname}
}
\makeatother

\newif\ifShowComments
\ShowCommentstrue
\def\strutdepth{\dp\strutbox}
\def\druk#1{\strut\vadjust{\kern-\strutdepth
        {\vtop to \strutdepth{%
                \baselineskip\strutdepth\vss
                        \llap{\hbox{#1}\quad}\null}}}}

\title{\bf
Bias in Gini coefficient estimation for gamma mixture populations
}

\author{
\text{Roberto Vila}$^{1}$
\,\,\,and
\text{Helton Saulo}$^{1,2}$\thanks{Corresponding author: H. Saulo, email: {heltonsaulo@gmail.com}
\newline
}
\\
{\small $^{1}$ Department of Statistics, University of Brasilia, Brasilia, Brazil}\\
{\small $^{2}$ Department of Economics, Federal University of Pelotas, Pelotas, Brazil}\\
}
\begin{document}
	\maketitle 	
	\begin{abstract}
This paper examines the properties of the Gini coefficient estimator for gamma mixture populations and reveals the presence of bias. In contrast, we show that sampling from a gamma distribution yields an unbiased estimator, consistent with prior research \citep{Baydil2025}. We derive an explicit bias expression for the Gini coefficient in gamma mixture populations, which serves as the foundation for proposing a bias-corrected Gini estimator. We conduct a Monte Carlo simulation study to evaluate the behavior of the bias-corrected Gini estimator.
	\end{abstract}
%
	
\section{Introduction}
	\noindent

The Gini coefficient is a widely used measure of income inequality and dispersion in a population \citep{Dorfman1979}.  The degree of inequality within a distribution can be quantified using this method, which has applications in a variety of fields, including economics, finance, ecology, health, and environment, among others; see, for example, \cite{Damgaard2000}, \cite{YAO01101999}, \cite{SUN2010601}, and \cite{Kharazmi2023}. According to \cite{Baydil2025}, the Gini coefficient is commonly utilized by the World Bank in order to evaluate the degree of economic disparity that exists across countries. This highlights the practical significance of the Gini coefficient.

Mathematically, the Gini coefficient is commonly defined in terms of the expected absolute differences between two independent, identically distributed (i.i.d.) random variables. A widely used estimator for this index is the sample Gini coefficient; see \cite{Deltas2003}. However, this estimator tends to be biased downward, particularly in small samples and across various distributions, including uniform, log-normal, and exponential models. To address this issue, \cite{Deltas2003} introduced the upward-adjusted sample Gini coefficient, which was shown to be unbiased when the population followed an exponential distribution. Recently, \cite{Baydil2025} extended the unbiasedness of the upward-adjusted sample Gini coefficient when the population follows a gamma distribution.

The gamma distribution is particularly appealing in economic modeling due to its flexibility in capturing mid-range income distributions, as opposed to the log-normal model, which is often used for high-income distributions; see \cite{Salem1974}. \cite{McDonald1979} derived the population Gini coefficient for a gamma-distributed population and explored its estimation via maximum likelihood and moment-based methods. Nevertheless, a single gamma distribution for modeling income may not be adequate as real-world income data is often heterogeneous, consisting of multiple economic classes (e.g., low-, middle-, and high-income earners). In contrast, a gamma mixture model can be a more flexible parametric model for an income distribution as studied by \cite{Chotikapanich2008}. In fact, mixtures offer the benefit of a flexible functional structure while maintaining the convenience of parametric models that facilitate statistical inference; see \cite{Chotikapanich2008}.

In this paper, we examine the estimation of the Gini coefficient in gamma mixture populations, a more general class of distributions that accounts for heterogeneity in income data. While the sample Gini coefficient is known to be an unbiased estimator for the population Gini coefficient in single gamma distributions \citep{Baydil2025}, we show that this property does not hold for gamma mixture models, leading to a systematic bias in estimation. Then, we derive an explicit expression for the bias and propose an unbiased Gini estimator when the population is gamma mixture distributed.

The rest of this paper unfolds as follows. In Section \ref{sec:02}, we present the theoretical foundations and key definitions. In Section \ref{sec:03}, we derive a closed-form expression for the expectation of the sample Gini coefficient estimator. In Section \ref{sec:04}, we demonstrate the bias introduced when estimating the Gini coefficient from gamma mixture distributions, providing both theoretical and empirical evidence. In Section \ref{sec:05}, we explore an Illustrative Monte Carlo simulation study. Finally, in Section \ref{sec:06}, we provide some concluding remarks.



\section{Preliminary results and some definitions}	
\label{sec:02}


This section provides the theoretical foundation for our analysis, introducing preliminary results and definitions. A key result, Proposition \ref{corollary-2}, presents an explicit formula for the Gini coefficient, a widely used statistical measure of income inequality \citep{Gini1936}.




The theoretical results presented in this section, pertaining to (almost surely) positive random variables, exhibit universal validity and applicability, transcending specific distributional forms and accommodating real-valued random variables with diverse support structures.

\begin{lemma}\label{main-theorem}
	Let $X_1$ and $X_2$ be two independent copies of a
(almost surely)	positive random variable $X$ with finite integral and common cumulative distribution function $F$ and let $g$ be a positive  real-valued integrable function of two (positive) variables. Then, the following identity holds:
	\begin{align*}
		\mathbb{E}[ \vert X_1-X_2\vert g(X_1,X_2)]
		=
		2\mathbb{E}(X)\left\{
		\mathbb{E}\left[g(X^*,X) \mathds{1}_{\{X<X^*\}}\right]
		-
		\mathbb{E}\left[g(X,X^*)\mathds{1}_{\{X^*<X\}}\right]	
		\right\},
	\end{align*}
	where $\mathds{1}_A$ is the indicator function of an event $A$,
	$X^*$ is independent of $X$ and has length-biased distribution, that is, its cumulative distribution function is given by the following Lebesgue-Stieltjes integral:
	\begin{align}\label{de-cdf-X*}
		F_{X^*}(u)
		=
		{1\over \mathbb{E}(X)}\, 
		{\int_{0}^u t{\rm d}F(t)},
		\quad u>0.
	\end{align}
\end{lemma}	
\begin{proof}
	From independence of $X_1$ and $X_2$, it follows that
	\begin{align}
		\mathbb{E}[ \vert X_1-&X_2\vert g(X_1,X_2)]
		\nonumber
		\\[0,2cm]
		&=
		\int_{0}^{\infty}
		\left[
		\int_{u}^{\infty}
		(v-u)g(u,v)
		{\rm d}F(v)
		\right]
		{\rm d}F(u)
		+
		\int_{0}^{\infty}
		\left[
		\int_{0}^{u}
		(u-v)g(u,v)
		{\rm d}F(v)
		\right]
		{\rm d}F(u)
		\nonumber
		\\[0,2cm]
		&=
		2
		\int_{0}^{\infty}
		u
		\left[
		\int_{0}^{u}
		g(u,v)
		{\rm d}F(v)
		\right]
		{\rm d}F(u)
		-
		2
		\int_{0}^{\infty}
		\left[
		\int_{0}^{u}
		vg(u,v)
		{\rm d}F(v)
		\right]
		{\rm d}F(u), \label{def-1}
	\end{align}
	where in the last identity we use the fact that $X_1$ and $X_2$ are identically distributed. 
	By using the definition \eqref{de-cdf-X*} of $F_{X^*}$, the  expression in \eqref{def-1} can be written as
	\begin{align*}
		&2\mathbb{E}(X)
		\int_{0}^{\infty}
		\left[
		\int_{0}^{u}
		g(u,v)
		{\rm d}F(v)
		\right]
		{\rm d}F_{X^*}(u)
		-
		2\mathbb{E}(X)
		\int_{0}^{\infty}
		\left[
		\int_{0}^{u}
		g(u,v)
		{\rm d}F_{X^*}(v)
		\right]
		{\rm d}F(u).
	\end{align*}
	Hence the proof readily follows.
\end{proof}

\begin{proposition}\label{main-theorem-1}
	Under the conditions of Lemma \ref{main-theorem}, if  $g$ is symmetric, that is, $g(u,v)=g(v,u)$, $u,v>0$, then
	\begin{align*}
		\mathbb{E}[ \vert X_1-X_2\vert g(X_1,X_2)]
		=
		2\mathbb{E}(X)\left\{
		2\mathbb{E}\left[g(X^*,X) \mathds{1}_{\{X<X^*\}}\right]
		-
		\mathbb{E}\left[g(X^*,X)\right] 	
		+
		\mathbb{E}
		\left[g(X^*,X)\mathds{1}_{\{X=X^*\}}\right] 	
		\right\}.
	\end{align*}
\end{proposition}
\begin{proof}
	The proof is immediate, therefore omitted.
\end{proof}

\begin{proposition}\label{corollary-1}
	Under the conditions of Proposition \ref{main-theorem-1}, we have
	\begin{align}\label{idd-1}
		\mathbb{E}(\vert X_1-X_2\vert)
		=
		2\mathbb{E}(X)\left[
		2\mathbb{P}(X<X^*)
		-
		1
		+
		\mathbb{P}(X=X^*)
		\right].
	\end{align}
	Equivalently,
	\begin{align}\label{idd-2}
		\mathbb{E}(\vert X_1-X_2\vert)
		=
		2\mathbb{E}(X)
		\left[
		2F_{X\over X+X^*}\left({1\over 2}\right)
		-
		1
		+
		\mathbb{P}(X=X^*)
		\right],
	\end{align}
	with $F_{X/(X+X^*)}$ being the cumulative distribution function of the independent ratio ${X/(X+X^*)}$.
\end{proposition}
\begin{proof}
	By taking, in Lemma \ref{main-theorem}, $g$ as the unitary constant function, we obtain \eqref{idd-1}. Identity in \eqref{idd-2} follows from the equality of events $\{X<X^*\}$ and $\{2X<X+X^*\}$.
\end{proof}

The Gini coefficient is defined using the Lorenz curve, which plots cumulative income share against population percentage. The line of perfect equality is a 45-degree line. The Gini coefficient is calculated as the ratio of the area between the Lorenz curve and the line of equality to the total area under the line. To derive explicit expressions, we use the standard definition based on the mean difference.
\begin{definition}
	The Gini coefficient \citep{Gini1936} of a random variable $X$ with finite mean $\mathbb{E}(X)$ is defined as
\begin{align}\label{Gini-coefficient}
	G={1\over 2}\, {\mathbb{E}(\vert X_1-X_2\vert)\over\mathbb{E}(X)},
\end{align}
where $X_1$ and $X_2$ are independent copies of $X$.
\end{definition}

\begin{proposition}\label{corollary-2}
	Under the conditions of Proposition \ref{corollary-1}, the Gini coefficient for $X$ is given by
	\begin{align*}
		G
		=
		2F_{X\over X+X^*}\left({1\over 2}\right)
		-
		1
		+
		\mathbb{P}(X=X^*),
	\end{align*}
	with $F_{X/(X+X^*)}$ being the cumulative distribution function of the independent ratio ${X/(X+X^*)}$ and $X^*$ distributed according to \eqref{de-cdf-X*}.
\end{proposition}

We proceed by defining the gamma mixture distribution for the case of a constant rate parameter, $\lambda > 0$. While our results can be extended to the more general case of non-constant rate parameters, our primary goal is to investigate the bias in the estimation of the Gini coefficient for a gamma mixture distribution, and we demonstrate that this can be achieved by focusing on the constant $\lambda$ scenario.
\begin{definition}\label{def-mixture}
	A random variable $X$ has a gamma mixture distribution \citep{Kitani2024} with parameter vector $\boldsymbol{\theta}=(\pi_1,\ldots,\pi_m,\alpha_1,\ldots,\alpha_m,\lambda)^\top$, denoted by $X\sim \text{GM}(\boldsymbol{\theta})$, if it has the following mixture density:
	\begin{align*}
		f(x;\boldsymbol{\theta})
		=
		\sum_{i=1}^{m}
		\pi_i g_i(x;\alpha_i,\lambda), 
		\quad x>0,
	\end{align*}
	where $m\in\mathbb{N}$ denotes the number of components, $\pi_i$ denotes the mixing proportion such that $\sum_{i=1}^{m}
	\pi_i=1$, $\pi_i> 0$, and	$g_i$ is the density of the gamma distribution with shape parameter $\alpha_i > 0$ and rate parameter $\lambda > 0$, that is,
	\begin{align}\label{def-g-pdf}
		g_i(x;\alpha_i,\lambda)
		=
		{\lambda^{\alpha_i}\over\Gamma(\alpha_i)} \, x^{\alpha_i-1}\exp(-\lambda x), 
		\quad x>0,
	\end{align}
	where $\Gamma(\cdot)$ denotes the (complete) gamma function.
\end{definition}

The gamma mixture distribution is a comprehensive framework that encompasses a broad range of important distributions, as illustrated in Figure \ref{Figure-1} \citep[adaptation of Figure 1 of ][]{Kitani2024}.
\begin{figure}[htb!]
	\vspace{-0.25cm}
	\centering
	{\includegraphics[height=12.0cm,width=16.0cm]{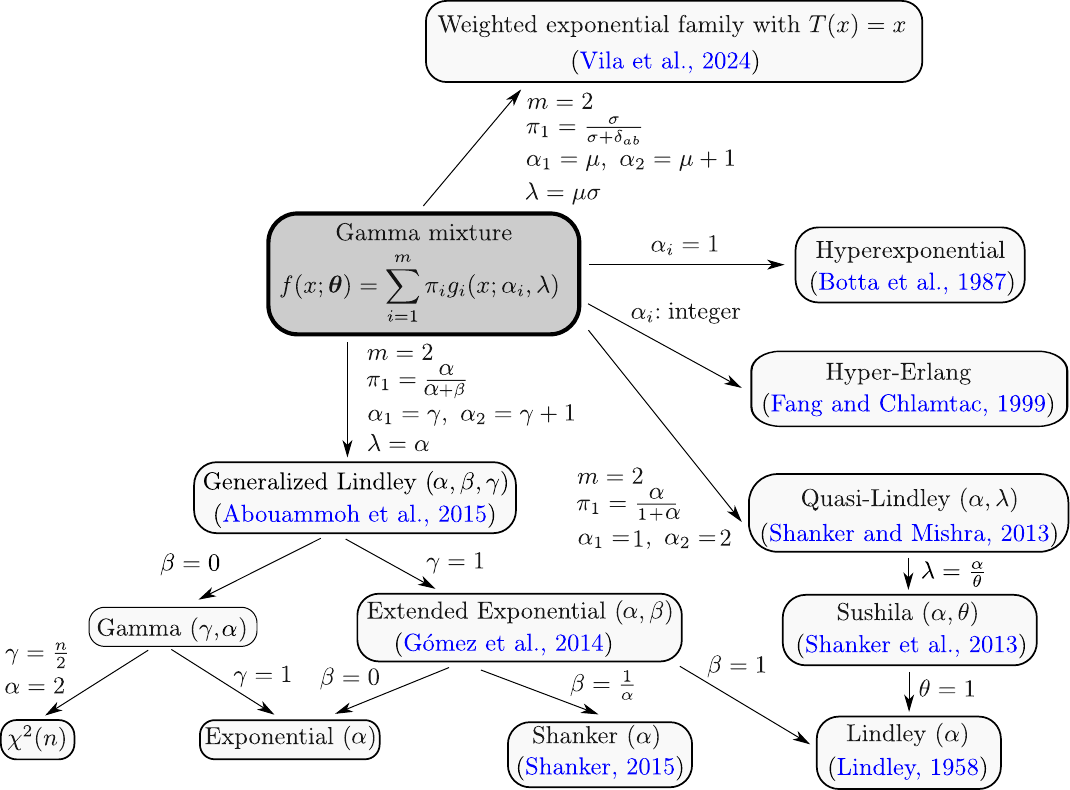}}
	\vspace{-0.2cm}
	\caption{Gamma mixture-type distribution relationships.}
	\label{Figure-1}
\end{figure}

\begin{example}\label{gm-example}
	If $X\sim \text{GM}(\boldsymbol{\theta})$, then $X^*\sim \text{GM}(\boldsymbol{\theta}^*)$, where $\boldsymbol{\theta}^*=(\pi^*_1,\ldots,\pi^*_m,\alpha_1+1,\ldots,\alpha_m+1,\lambda)^\top$, and
	\begin{align}\label{pi-star}
		\pi^*_j
		\equiv
		{
			\pi_j \alpha_j
			\over 		
			\sum_{k=1}^{m} {\pi_k\alpha_k}
		},
		\quad 
		j=1,\ldots,m.
	\end{align}

	As $X$ and $X^*$ are independent, by Jacobian method the density of the ratio ${X/(X+X^*)}$ satisfies:
	\begin{align}\label{id-inte}
		f_{{X\over X+X^*}}(z)
		&=
		(s+1)^2\int_{0}^{\infty}xf(x;\boldsymbol{\theta})f(sx;\boldsymbol{\theta}^*){\rm d}x,
		\quad 
		s={1\over z}-1,
		\ 0<z<1.
	\end{align}
	By using the definitions of $f(x;\boldsymbol{\theta})$ and $f(sx;\boldsymbol{\theta}^*)$ in \eqref{id-inte},  
	\begin{align*}
		f_{{X\over X+X^*}}(z)
		&=
		(s+1)^2
		\sum_{i,j=1}^{m}
		\pi_i \pi_j^*
		{\lambda^{\alpha_i+\alpha_j+1} \over\Gamma(\alpha_i) \Gamma(\alpha_j+1)} \,
		s^{\alpha_j}
		\int_{0}^{\infty}  
		x^{\alpha_i+\alpha_j}
		\exp[-\lambda(1+ s)x]
		{\rm d}x
		\\[0,2cm]
		&=
		\sum_{i,j=1}^{m}
		\pi_i \pi_j^*
		\,
		{z^{\alpha_i-1}(1-z)^{\alpha_j}\over \text{B}(\alpha_i, \alpha_j+1)},
		\quad 
		0<z<1,
	\end{align*}
	where  $\text{B}(\cdot,\cdot)$ is the (complete)  beta function.
	That is, the distribution of ${X/(X+X^*)}$ is a mixture of beta distributions with parameters $\alpha_i$ and $\alpha_j+1$, $i,j=1,\ldots,m$.
	Then, the cumulative distribution function of ${X/(X+X^*)}$ can be written as 
	\begin{align*}
		F_{{X\over X+X^*}}(z)
		=
		\sum_{i,j=1}^{m}
		\pi_i \pi_j^*
		\,
		{\text{B}(z;\alpha_i,\alpha_j+1)\over \text{B}(\alpha_i, \alpha_j+1)}
		=
		\sum_{i,j=1}^{m}
		\pi_i \pi_j^*
		\,
		{_2F_1\left(\alpha_i,-\alpha_j;\alpha_i+1;z\right)\over \alpha_i \text{B}(\alpha_i, \alpha_j+1)} \,z^{\alpha_i},
		\quad 
		0<z<1,
	\end{align*}
	where we have used the well-known identity $\text{B}(x;a,b)=(x^a/a) \,_2F_1(a,1-b;a+1;x)$ that relates the incomplete beta function $\text{B}(x;a,b)$ to the hypergeometric function
	$_2F_1(a,b;c;x)$.
	Applying Proposition \ref{corollary-2}, the Gini coefficient of the GM distribution can be expressed as \begin{align}\label{Gini-GM}
		G=
		\sum_{i,j=1}^{m}
		\pi_i \pi_j^*
		\,
		{_2F_1\left(\alpha_i,-\alpha_j;\alpha_i+1;{1\over 2}\right)\over 2^{\alpha_i-1}\alpha_i \text{B}(\alpha_i, \alpha_j+1)} 
		-
		1,
	\end{align}
	because $\mathbb{P}(X=X^*)=0$.
\end{example}

\begin{example}
	By taking $m=1$, $\pi_1=\pi_1^*=1$ and $\alpha_1=\alpha_2=\alpha$  in \eqref{Gini-GM}, we have
	\begin{align}\label{gini-1}
		G=
		{_2F_1\left(\alpha,-\alpha;\alpha+1;{1\over 2}\right)\over 2^{\alpha-1}\alpha B(\alpha, \alpha+1)}  
		-
		1.
	\end{align}
	By using the identity \citep{WolframResearch2024}:
	\begin{align*}
		\,_{2}F_{1}\left(a,-a;c;{1\over 2}\right)
		=
		{\sqrt{\pi}\Gamma(c)\over 2^c}
		\left[
		{1\over\Gamma({a+c+1\over 2})\Gamma({c-a\over 2})}
		+
		{1\over\Gamma({a+c\over 2})\Gamma({c-a+1\over 2})}
		\right],
	\end{align*}
	the Gini coefficient \eqref{gini-1} becomes
	\begin{align}\label{gini-2}
		G
		=
		{\Gamma(2\alpha+1)\over \alpha^2 2^{2\alpha} \Gamma^2(\alpha)} 
		+
		{2^{1-2\alpha}\sqrt{\pi}\Gamma(2\alpha)\over\Gamma(\alpha)\Gamma(\alpha+{1\over 2})} 
		-
		1
		=
		{\Gamma(2\alpha+1)\over \alpha^2 2^{2\alpha} \Gamma^2(\alpha)},
	\end{align}
	where in the last equality we have used the 
	Legendre duplication formula \citep{Abramowitz1972}:
	$
	\Gamma(x)\Gamma\left(x+{1/ 2}\right)=2^{1-2x}\sqrt{\pi}\Gamma(2x).
	$
	Again, by applying Legendre duplication formula in \eqref{gini-2}, we get the following expression
	\begin{align*}
		G
		=
		{\Gamma(\alpha+{1\over 2})\over \sqrt{\pi} \alpha\Gamma(\alpha)},
	\end{align*}
	which is consistent with the well-known formula for the Gini coefficient of the gamma distribution, as reported in the existing literature \citep[see, for example,][]{McDonald1979}.
\end{example}

\section{The main result}\label{sec:03}


Utilizing the foundational results presented in Section \ref{sec:02}, is this part, we obtain a simple closed-form expression (Theorem \ref{the-main-2}) for the expected value of the Gini coefficient estimator $\widehat{G}$, initially proposed by \cite{Deltas2003},
\begin{align}\label{gini-estimadtor-def}
	\widehat{G}
	=
	{1\over n-1} 
	\left[\dfrac{\displaystyle\sum_{1\leqslant i<j\leqslant n}\vert X_i-X_j\vert}{\displaystyle\sum_{i=1}^{n} X_i}\right],
	\quad 
	n\in\mathbb{N}, \, n\geqslant 2,
\end{align}
where $X_1, X_2,\ldots,X_n$ are i.i.d. observations from the population.

The following theorem is valid only for (almost surely) positive random variables and its proof adapts similar technical steps as reference \cite{Baydil2025}.
\begin{theorem}\label{the-main-2}
	Let $X_1,X_2,\ldots$ be independent copies of a (almost surely)
	positive random variable $X$ with finite integral and common cumulative distribution function $F$.
	The following holds: 
	\begin{align*}
		\mathbb{E}(\widehat{G})
		=
		n
		\mathbb{E}(X)
		\left\{2R_{1}(F)-R_{\infty}(F)		+
		\mathbb{E}
		\left[g(X^*,X)\mathds{1}_{\{X=X^*\}}\right]\right\},
		\quad 
		n\in\mathbb{N}, \, n\geqslant 2,
	\end{align*}
	where $X^*$ has length-biased distribution (see Lemma \ref{main-theorem}), 
	$R_{1}(F)\equiv \lim_{\varepsilon\to 1}R_{\varepsilon}(F)$, 
	$R_{\infty}(F)\equiv \lim_{\varepsilon\to \infty}R_{\varepsilon}(F)$,
	\begin{align}\label{def-R-function}
		&R_{\varepsilon}(F)
		\equiv
		\int_{0}^{\infty}
		\mathbb{E}\left[
		\exp\left(-xX^*\right)
		H(x,\varepsilon X^*)
		\right]
		\mathscr{L}^{n-2}_F(x){\rm d}x,
		\quad 
		\varepsilon>0,
		\\[0,2cm]
		&
		H(x,x^*)
		\equiv
		\int_{0}^{x^*}
		\exp\left(-x u\right) 
		{\rm d}F(u),
		\quad x, x^*>0,
		\label{def-H-function}
		\\[0,2cm]
		&
		{g}(u,v)
		\equiv 
		\int_{0}^{\infty}\exp\left\{-(u+v) x\right\} \mathscr{L}^{n-2}_F(x){\rm d}x,
		\quad u,v>0,
		\label{def-g}
	\end{align}
	and
	$\mathscr{L}_F(p)=\int_{0}^{\infty}\exp(-pu) {\rm d}F(u)$ is the Laplace transform corresponding to distribution $F$.
	In the above, we are assuming that the Lebesgue-Stieltjes integrals and improper integrals involved exist.
\end{theorem}
\begin{proof}
		%
		%
		%
	%
	By using the very well-known identity
	\begin{align*}
		{z}
		\int_{0}^{\infty}\exp(-z x){\rm d}x
		=
		1,
	\end{align*}
	with $z=\sum_{i=1}^{n} X_i$, and by taking advantage of the independent and identically distributive nature of $X_1,X_2,\ldots$, we have
	\begin{align}
		\mathbb{E}\left[\dfrac{\displaystyle\sum_{1\leqslant i<j\leqslant n}\vert X_i-X_j\vert}{\displaystyle\sum_{i=1}^{n} X_i}\right]
		&=
		\mathbb{E}\left[\sum_{1\leqslant i<j\leqslant n}\vert X_i-X_j\vert \int_{0}^{\infty}\exp\left\{-\left(\displaystyle\sum_{i=1}^{n} X_i\right) x\right\}{\rm d}x\right]
		\nonumber
		\\[0,2cm]
		&=
		\mathbb{E}\left[\sum_{1\leqslant i<j\leqslant n}\vert X_i-X_j\vert \int_{0}^{\infty}\exp\left\{-(X_1+X_2) x\right\}\exp\left\{-\left(\displaystyle\sum_{i=3}^{n} X_i\right) x\right\}{\rm d}x\right]
		\nonumber
		\\[0,2cm]
		&=
		\binom{n}{2}
		\mathbb{E}\left[
		\vert X_1-X_2\vert 
		\int_{0}^{\infty}\exp\left\{-(X_1+X_2) x\right\} \mathscr{L}^{n-2}_F(x){\rm d}x
		\right].
		\label{eq-1}
	\end{align}
	By taking ${g}(u,v)$ as in \eqref{def-g},
	the identity in \eqref{eq-1} can be written as
	\begin{align*}
		\mathbb{E}\left[\dfrac{\displaystyle\sum_{1\leqslant i<j\leqslant n}\vert X_i-X_j\vert}{\displaystyle\sum_{i=1}^{n} X_i}\right]
		=
		\binom{n}{2}
		\mathbb{E}\left[
		\vert X_1-X_2\vert 
		{g}(X_1,X_2)
		\right].
	\end{align*}
	As ${g}$ in \eqref{def-g} is symmetric, by applying Proposition \ref{main-theorem-1}, the above identity becomes
	\begin{align}\label{exp-end}
		\mathbb{E}(\widehat{G})
		=
		n\mathbb{E}(X)\left\{
		2\mathbb{E}\left[g(X^*,X) \mathds{1}_{\{X<X^*\}}\right]
		-
		\mathbb{E}\left[g(X^*,X)\right]	
		+
		\mathbb{E}
		\left[g(X^*,X)\mathds{1}_{\{X=X^*\}}\right] 
		\right\}.
	\end{align}

	Now, note that
	\begin{align}\label{eq-1-1}
		\mathbb{E}\left[{g}(X^*,X) \mathds{1}_{\{X<X^*\}}\right]
		&=
		\int_{0}^{\infty}
		\mathbb{E}\left[
		\exp\left\{-(X^*+X) x\right\} 
		\mathds{1}_{\{X<X^*\}}\right]
		\mathscr{L}^{n-2}_F(x){\rm d}x \nonumber
		\\[0,2cm]
		&=
		\int_{0}^{\infty}
		\mathbb{E}\left[
		\exp\left(-xX^*\right)
		\int_{0}^{\infty}
		\exp\left(-x u\right) 
		\mathds{1}_{\{u<X^*\}}
		{\rm d}F(u)
		\right]
		\mathscr{L}^{n-2}_F(x){\rm d}x
	\end{align}
	and
	\begin{align} \label{eq-2}
		\mathbb{E}\left[{g}(X^*,X)\right]
		&=
		\int_{0}^{\infty}
		\mathbb{E}\left[
		\exp\left\{-(X^*+X) x\right\} 
		\right]
		\mathscr{L}^{n-2}_F(x){\rm d}x \nonumber
		\\[0,2cm]
		&=
		\int_{0}^{\infty}
		\mathbb{E}\left[
		\exp\left(-xX^*\right) 
		\int_{0}^{\infty}
		\exp\left(-x u\right) 
		{\rm d}F(u)
		\right]
		\mathscr{L}^{n-2}_F(x){\rm d}x.
	\end{align}
	By replacing \eqref{eq-1-1} and \eqref{eq-2} in formula \eqref{exp-end}, and by using the notations in \eqref{def-R-function} and \eqref{def-H-function}, the proof of the theorem readily follows.
\end{proof}

\begin{remark}\label{remark-main}
	Note that the Laplace transform corresponding to distribution $F$ can be obtained from $H(\cdot,\cdot)$ in \eqref{def-H-function} as follows:
	\begin{align*}
		\mathscr{L}_F(x)
		=
		\lim_{\varepsilon\to\infty}
		H(x,\varepsilon x^*), \quad x, x^*>0.
	\end{align*}
\end{remark}

As an immediate consequence of Theorem \ref{the-main-2}, the following  result follows.
\begin{corollary}\label{the-main-3}
	Under the conditions of Theorem \ref{the-main-2}, with $X$ being a (almost surely) positive, absolutely continuous random variable with finite integral and common cumulative distribution function $F$, we have
	\begin{align*}
		\mathbb{E}(\widehat{G})
		=
		n
		\mathbb{E}(X)
		\left[2R_{1}(F)-R_{\infty}(F)\right],
		\quad 
		n\in\mathbb{N}, \, n\geqslant 2.
	\end{align*}
\end{corollary}

		\section{The Gini coefficient estimator is biased for gamma mixtures}\label{sec:04}

		In this section, we show that the estimator  $\widehat{G}$ in \eqref{gini-estimadtor-def} of the Gini coefficient $G$  (given in Example \ref{gm-example}) for a gamma mixture population is biased, whereas sampling from a gamma-distributed population eliminates this bias, as previously established by \cite{Baydil2025}.

		Indeed, if $X\sim \text{GM}(\boldsymbol{\theta})$ (see Definition \ref{def-mixture}), then, $H(\cdot,\cdot)$ in \eqref{def-H-function} can be written as
		%
		%
		%
		%
		\begin{align}\label{mix-1}
			H(x,\varepsilon x^*)
			&=
			\sum_{i=1}^{m}		
			\pi_i 
			\int_{0}^{\varepsilon x^*}
			\exp\left(-x u\right) 
			g_i(u;\alpha_i,\lambda)
			{\rm d}u,
			\quad x, x^*,\varepsilon>0,
			\nonumber
			\\[0,2cm]
			&=
			\sum_{i=1}^{m}		
			\pi_i \,
			{\lambda^{\alpha_i}\over\Gamma(\alpha_i)} 
			\int_{0}^{\varepsilon x^*}
			u^{\alpha_i-1}
			\exp[-(x+\lambda)u]
			{\rm d}u
			\nonumber
			\\[0,2cm]
			&=
			\sum_{i=1}^{m}		
			\pi_i \,
			\left({\lambda\over x+\lambda}\right)^{\alpha_i}
			{\gamma(\alpha_i,(x+\lambda)\varepsilon x^*)\over \Gamma(\alpha_i)},
		\end{align}	
		with $g_i(\cdot;\alpha_i,\lambda)$ being as in \eqref{def-g}.
		By using Remark \ref{remark-main}, from \eqref{mix-1}, we obtain
		\begin{align}\label{laplace-transform}
			\mathscr{L}_F(x)
			=
			\lim_{\varepsilon\to\infty}
			H(x,\varepsilon x^*)
			=
			\sum_{i=1}^{m}		
			\pi_i \,
			\left({\lambda\over x+\lambda}\right)^{\alpha_i},
			\quad x, x^*>0,
		\end{align}
		where the well-known identity $\Gamma(x)=\lim_{y\to\infty}\gamma(x,y)$ has been used. Furthermore, since $X^*\sim \text{GM}(\boldsymbol{\theta}^*)$ with $\boldsymbol{\theta}^*$ being as in Example \ref{gm-example}, by using \eqref{mix-1}, we have
		\begin{align*}
			\mathbb{E}[
			\exp\left(-xX^*\right) &
			H(x,\varepsilon X^*)
			]
			=
			\sum_{i=1}^{m}		
			\pi_i \,
			\left({\lambda\over x+\lambda}\right)^{\alpha_i}
			{1\over \Gamma(\alpha_i)} \,
			\mathbb{E}\left[
			\exp\left(-xX^*\right)
			\gamma(\alpha_i,(x+\lambda)\varepsilon X^*)
			\right]
			\\[0,2cm]
			&=
			\sum_{i,j=1}^{m}		
			\pi_i \pi_j^*\,
			\left({\lambda\over x+\lambda}\right)^{\alpha_i}
			{\lambda^{\alpha_j+1}\over \Gamma(\alpha_i)\Gamma(\alpha_j+1)}
			\int_0^\infty
			v^{\alpha_j}
			\exp\left[-(x+\lambda) v\right]
			\gamma(\alpha_i,(x+\lambda)\varepsilon v)
			{\rm d}v
			\\[0,2cm]
			&=
			\sum_{i,j=1}^{m}		
			\pi_i \pi_j^*\,
			{\lambda^{\alpha_j+1}\over \alpha_i B(\alpha_i,\alpha_j+1)} 
			\left({\lambda\over x+\lambda}\right)^{\alpha_i}
			{(x+\lambda)^{\alpha_i}\varepsilon^{\alpha_i}\over [(x+\lambda)\varepsilon+x+\lambda]^{\alpha_i+\alpha_j+1}}
			\\[0,2cm]
			&\times 
			\,_{2}F_{1}\left(\alpha_i+\alpha_j+1,1;\alpha_i+1;{\varepsilon\over \varepsilon+1}\right),
		\end{align*}	
		where 
		in the last line the identity \citep{DAurizio2016}:
		\begin{align*}
			\int_0^\infty 
			x^{a-1}
			\exp(-sx)\gamma(b,\theta x){\rm d}x
			=
			{\theta^b\Gamma(a+b)\over b(s+\theta)^{a+b}}
			\,_{2}F_{1}\left(a+b,1;b+1;{\theta\over s+\theta}\right),
		\end{align*}
		has been used.
		
		Hence, $R_{\varepsilon}(F)$ in \eqref{def-R-function} can be written as
		\begin{align}\label{r-varepsilon}
			R_{\varepsilon}(F)
			&=
			\int_{0}^{\infty}
			\mathbb{E}\left[
			\exp\left(-xX^*\right)
			H(x,\varepsilon X^*)
			\right]
			\mathscr{L}^{n-2}_F(x){\rm d}x
			\nonumber
			\\[0,2cm]
			&=
			\sum_{i,j=1}^{m}		
			\pi_i \pi_j^*\,
			{\lambda^{\alpha_j+1}\over \alpha_i B(\alpha_i,\alpha_j+1)} 
			\,_{2}F_{1}\left(\alpha_i+\alpha_j+1,1;\alpha_i+1;{\varepsilon\over \varepsilon+1}\right)
			\nonumber
			\\[0,2cm]
			&\times 
			\int_{0}^{\infty}
			{{\lambda}^{\alpha_i} \varepsilon^{\alpha_i}\over [(x+\lambda)\varepsilon+x+\lambda]^{\alpha_i+\alpha_j+1}} 
			\mathscr{L}^{n-2}_F(x)
			{\rm d}x.
		\end{align}
		Furthermore, by using, in \eqref{r-varepsilon}, the Euler's Hypergeometric transformation
		\citep{Abramowitz1972}:
		${}_{2}F_{1}(a,b;c;x)=(1-x)^{c-a-b}{}_{2}F_{1}(c-a,c-b;c;x)$, we have 
		\begin{align}\label{r-varepsilon-1}
			R_{\varepsilon}(F)
			=
			\sum_{i,j=1}^{m}		
			\pi_i \pi_j^*\,
			{\, 
				{}_{2}F_{1}\left(-\alpha_j,\alpha_i;\alpha_i+1;{\varepsilon\over \varepsilon+1}\right)\over \alpha_i B(\alpha_i,\alpha_j+1)} 
			\left({\varepsilon\over\varepsilon+1}\right)^{\alpha_i}
			\int_{0}^{\infty}
			\left({\lambda\over x+\lambda}\right)^{\alpha_i+\alpha_j+1} 
			\mathscr{L}^{n-2}_F(x)
			{\rm d}x.
		\end{align}

		%
		
		Now, by taking $\varepsilon\to 1$ in  \eqref{r-varepsilon-1},  from formula \eqref{laplace-transform} of $\mathscr{L}_F(x)$,
		we get
		\begin{align}\label{R1-1}
			&R_{1}(F)
			=
			\sum_{i,j=1}^{m}		
			\pi_i \pi_j^*\,
			{				{}_{2}F_{1}\left(-\alpha_j,\alpha_i;\alpha_i+1;{1\over 2}\right)\over 2^{\alpha_i}\alpha_i B(\alpha_i,\alpha_j+1)} 
			\int_{0}^{\infty}
			\left({\lambda\over x+\lambda}\right)^{\alpha_i+\alpha_j+1}
			\mathscr{L}^{n-2}_F(x)
			{\rm d}x
			\nonumber
			\\[0,2cm]
			&=
			\sum_{i,j=1}^{m}		
			\pi_i \pi_j^*\,
			{				{}_{2}F_{1}\left(-\alpha_j,\alpha_i;\alpha_i+1;{1\over 2}\right)\over 2^{\alpha_i}\alpha_i B(\alpha_i,\alpha_j+1)} 
			\sum_{1\leqslant i_1,\ldots,i_{n-2}\leqslant m}	
			\pi_{i_1}\cdots\pi_{i_{n-2}}
			\int_{0}^{\infty}
			\left({\lambda\over x+\lambda}\right)^{\alpha_{i_1}+\cdots+ \alpha_{i_{n-2}}+\alpha_i+\alpha_j+1}
			{\rm d}x
			\nonumber
			\\[0,2cm]
			&=
			\lambda
			\sum_{i,j=1}^{m}		
			\pi_i \pi_j^*\,
			{				{}_{2}F_{1}\left(-\alpha_j,\alpha_i;\alpha_i+1;{1\over 2}\right)\over 2^{\alpha_i}\alpha_i B(\alpha_i,\alpha_j+1)} 
			\sum_{1\leqslant i_1,\ldots,i_{n-2}\leqslant m}	
			\pi_{i_1}\cdots\pi_{i_{n-2}}
			\,
			{1\over \alpha_{i_1}+\cdots+ \alpha_{i_{n-2}}+\alpha_i+\alpha_j}.
		\end{align}
		
		On the other hand, by setting $\varepsilon\to\infty$ in  \eqref{r-varepsilon-1}, from formula \eqref{laplace-transform} of $\mathscr{L}_F(x)$ and from definition \eqref{pi-star} of $\pi_j^*$, we obtain
		\begin{align}\label{Ri}
			R_{\infty}(F)	
			&=
			\sum_{i,j=1}^{m}		
			\pi_i \pi_j^*\,
			{ \,_{2}F_{1}\left(-\alpha_j,\alpha_i;\alpha_i+1;1\right) \over \alpha_i B(\alpha_i,\alpha_j+1)} 
			\int_{0}^{\infty}
			\left({\lambda\over x+\lambda}\right)^{\alpha_i+\alpha_j+1}
			\mathscr{L}^{n-2}_F(x)
			{\rm d}x
			\nonumber
			\\[0,2cm]
			&=
			{1\over \sum_{k=1}^{m}		
				{\pi_k\alpha_k}}
			\sum_{j=1}^{m}		
			\pi_j\alpha_j
			\int_{0}^{\infty}
			\left({\lambda\over x+\lambda}\right)^{\alpha_j+1}
			\mathscr{L}^{n-1}_F(x)
			{\rm d}x
			\nonumber
			\\[0,2cm]
			&=
			{1\over \sum_{k=1}^{m}		
				{\pi_k\alpha_k}}
			\sum_{j=1}^{m}		
			\pi_j\alpha_j
			\sum_{1\leqslant i_1,\ldots,i_{n-1}\leqslant m}	
			\pi_{i_1}\cdots\pi_{i_{n-1}}
			\int_{0}^{\infty}
			\left({\lambda\over x+\lambda}\right)^{\alpha_{i_1}+\cdots+\alpha_{i_{n-1}}+\alpha_j+1}
			{\rm d}x
			\nonumber
			\\[0,2cm]
			&=
			{1\over \sum_{k=1}^{m}		
				{\pi_k\alpha_k\over\lambda}}
			\sum_{1\leqslant i_1,\ldots,i_{n-1}\leqslant m}	
			\pi_{i_1}\cdots\pi_{i_{n-1}} \,
			\sum_{j=1}^{m}		
			{\pi_j\alpha_j 
				\over 
				\alpha_{i_1}+\cdots+\alpha_{i_{n-1}}+\alpha_j},
		\end{align}
		where in the second equality the identity $ \,_{2}F_{1}\left(-\alpha_j,\alpha_i;\alpha_i+1;1\right) = \alpha_i B(\alpha_i,\alpha_j+1)$ has been used.
		
		Then, by using \eqref{R1-1} and \eqref{Ri}, from Corollary \ref{the-main-3}, we get
		\begin{align}
			\mathbb{E}(\widehat{G})
			&=
			n
			\mathbb{E}(X)
			\left[2R_{1}(F)-R_{\infty}(F)\right]
			\nonumber
			\\[0,2cm]
			&=
			\sum_{i,j=1}^{m}		
			\pi_i \pi_j^*\,
			{				{}_{2}F_{1}\left(-\alpha_j,\alpha_i;\alpha_i+1;{1\over 2}\right)\over 2^{\alpha_i-1}\alpha_i B(\alpha_i,\alpha_j+1)} 
			\left[
			\sum_{1\leqslant i_1,\ldots,i_{n-2}\leqslant m}	
			\pi_{i_1}\cdots\pi_{i_{n-2}}
			\,
			{n\sum_{k=1}^{m}		
				{\pi_k\alpha_k}\over \alpha_{i_1}+\cdots+ \alpha_{i_{n-2}}+\alpha_i+\alpha_j}
			\right]
			\nonumber
			\\[0,2cm]
			&-
			\sum_{1\leqslant i_1,\ldots,i_{n-1}\leqslant m}	
			\pi_{i_1}\cdots\pi_{i_{n-1}} \,
			\sum_{j=1}^{m}		
			{n\pi_j\alpha_j 
				\over 
				\alpha_{i_1}+\cdots+\alpha_{i_{n-1}}+\alpha_j}.
				\label{exp-final}
		\end{align}

		Therefore, by \eqref{exp-final} and \eqref{Gini-GM}, the bias of $\widehat{G}$ relative to $G$, denoted by $\text{Bias}(\widehat{G},G)$, can be written as
		{\small
			\begin{align}\label{eq:bias_correc_gini}
				\text{Bias}(\widehat{G},G)
				&=
				\sum_{i,j=1}^{m}		
				\pi_i \pi_j^*\,
				{				{}_{2}F_{1}\left(-\alpha_j,\alpha_i;\alpha_i+1;{1\over 2}\right)\over 2^{\alpha_i-1}\alpha_i B(\alpha_i,\alpha_j+1)} 
				\left[
				\sum_{1\leqslant i_1,\ldots,i_{n-2}\leqslant m}	
				\pi_{i_1}\cdots\pi_{i_{n-2}}
				\,
				{n\sum_{k=1}^{m}		
					{\pi_k\alpha_k}\over \alpha_{i_1}+\cdots+ \alpha_{i_{n-2}}+\alpha_i+\alpha_j}
				-1
				\right]
				\nonumber
				\\[0,2cm]
				&
				+
				\left[
				1-
				\sum_{1\leqslant i_1,\ldots,i_{n-1}\leqslant m}	
				\pi_{i_1}\cdots\pi_{i_{n-1}} \,
				\sum_{j=1}^{m}		
				{n\pi_j\alpha_j 
					\over 
					\alpha_{i_1}+\cdots+\alpha_{i_{n-1}}+\alpha_j}
				\right].
			\end{align}
		}

		\begin{remark}
			Note that $\text{Bias}(\widehat{G},G)=0$ if and only $\alpha_1=\cdots=\alpha_m=\alpha$, that is, the estimator $\widehat{G}$ is unbiased  when sampling from a gamma distribution \citep{Baydil2025}.
		\end{remark}

\section{Illustrative simulation study}\label{sec:05}

Note that a bias-corrected Gini estimator can then be proposed from \eqref{Gini-GM} and \eqref{eq:bias_correc_gini} as
\begin{align}\label{Gini-GM-BC} \nonumber
\widehat{G}_{\text{bc}} &= \widehat{G} - \text{Bias}(\widehat{G}, G)\\[0,2cm] \nonumber
&=
	{1\over n-1}
	\left[\dfrac{\displaystyle\sum_{1\leqslant i<j\leqslant n}\vert X_i-X_j\vert}{\displaystyle\sum_{i=1}^{n} X_i}\right] \\ \nonumber
&-
	\sum_{i,j=1}^{m}
\widehat\pi_i \widehat\pi_j^*\,
{				{}_{2}F_{1}\left(-\widehat\alpha_j,\widehat\alpha_i;\widehat\alpha_i+1;{1\over 2}\right)\over 2^{\widehat\alpha_i-1}\widehat\alpha_i B(\widehat\alpha_i,\widehat\alpha_j+1)}
\left[
\sum_{1\leqslant i_1,\ldots,i_{n-2}\leqslant m}
\widehat\pi_{i_1}\cdots\widehat\pi_{i_{n-2}}
\,
{n\sum_{k=1}^{m}
	{\widehat\pi_k\widehat\alpha_k}\over \widehat\alpha_{i_1}+\cdots+ \widehat\alpha_{i_{n-2}}+\widehat\alpha_i+\widehat\alpha_j}
	-1
\right]
\\[0,2cm]
&
-
\left[
1-
\sum_{1\leqslant i_1,\ldots,i_{n-1}\leqslant m}
\widehat\pi_{i_1}\cdots\widehat\pi_{i_{n-1}} \,
\sum_{j=1}^{m}
{n\widehat\pi_j\widehat\alpha_j
	\over
	\widehat\alpha_{i_1}+\cdots+\widehat\alpha_{i_{n-1}}+\widehat\alpha_j}
	\right], 	\quad
	n\in\mathbb{N}, \, n\geqslant 2,
\end{align}
where hat notation on the mixing proportion and shape parameters denotes the maximum likelihood estimators. Here, we perform a Monte Carlo simulation to evaluate the behaviour of the bias-corrected Gini estimator in \eqref{Gini-GM-BC}. We consider a mixture of two gamma distributions with parameters
$(\pi_1,\pi_2,\alpha_1,\alpha_2,\lambda)$. The simulation scenario considers the following setting: sample size $n \in \{10, 11, 12, 13, 14, 15, 16\}$, mixing proportions $\pi_1=0.60$ and $p_2=0.40$, shape parameters $\alpha_1=0.5$ and
$\alpha_2\in\{0.5,1.0,2.0,3.0,5.0\}$, and the rate parameter is set to $\lambda = 1$. The steps of the Monte Carlo simulation study are described in Algorithm 1.

\begin{algorithm}
\caption{Monte Carlo simulation for bias-corrected Gini estimator.}
\begin{algorithmic}[1]
\State \textbf{Input:} Number of simulations $N_{\text{sim}}=100$, sample sizes $n \in \{10, \dots, 16\}$,
         true parameters $(\pi_1, \pi_2) = (0.60, 0.40)$, $\alpha_1=0.5$, $\alpha_2\in\{0.5,1.0,2.0,3.0,5.0\}$, and $\lambda = 1$.
\State \textbf{Output:} Gini estimate and its bias-corrected version.

\For{each sample size $n \in \{10, \dots, 16\}$ (or $\alpha_1=0.5$, $\alpha_2\in\{0.5,1.0,2.0,3.0,5.0\}$)}
    \For{each simulation run $s = 1, \dots, N_{\text{sim}}$}
        \State \textbf{Step 1:} Generate data
        \State Generate class labels $z \sim \text{Categorical}(\pi_1, \pi_2)$.
        \State Sample data $X$ from a gamma mixture:
        \[
        X_i \sim \sum_{k=1}^{2} \pi_k  \text{Gamma}(\alpha_k, \lambda).
        \]
        \State \textbf{Step 2:} Estimate gamma mixture parameters using the maximum likelihood method
        \State Compute $\widehat{\pi}_1$, $\widehat{\pi}_2$, $\widehat{\alpha}_1$, $\widehat{\alpha}_2$ and $\widehat{\lambda}$ by optimizing the likelihood function.
        \State \textbf{Step 3: Compute Gini coefficient estimates}
        \State Compute $\widehat{G}$ using the standard formula for the Gini coefficient \eqref{Gini-GM}.
        \State Compute bias correction term $\text{Bias}(\widehat{G}, G)$ using \eqref{eq:bias_correc_gini}.
        \State Compute bias-corrected estimate using \eqref{Gini-GM-BC}, that is,
        \[
        \widehat{G}_{\text{bc}} = \widehat{G} - \text{Bias}(\widehat{G}, G).
        \]
    \EndFor
    \State \textbf{Step 4: Compute Monte Carlo Averages}
    \State Compute mean of $\widehat{G}$ and $\widehat{G}_{\text{bc}}$.
\EndFor
\State \textbf{Return:} Averages of $\widehat{G}$ and $\widehat{G}_{\text{bc}}$.
\end{algorithmic}\label{algorithm:01}
\end{algorithm}

Figure \ref{fig:gini_comparison1} shows the average values of the standard and bias-corrected Gini coefficient estimates for different sample sizes $n$, holding $\alpha_1 = 0.5$ and $\alpha_2 = 2.0$ constant. From this figure, we observe that the standard Gini coefficient estimator ($\widehat{G}$, red line) tends to overestimate the Gini coefficient (positive bias). However, as the sample size increases, the bias diminishes slightly, demonstrating a convergence towards the true value, as expected. The bias-corrected estimator ($\widehat{G}_{\text{bc}}$, blue line), on the other hand, consistently produces lower estimates across all sample sizes.

Figure \ref{fig:gini_comparison2} presents the behavior of the Gini coefficient estimator as a function of $\alpha_2$, while keeping $\alpha_1 = 0.5$ and $n = 15$ fixed. From this figure, we observe that as $\alpha_2$ ($>0.5$) increases, the bias in the standard Gini estimator ($\widehat{G}$, red line) becomes more pronounced.

\begin{figure}[!ht]
    \centering
    \includegraphics[scale=0.7]{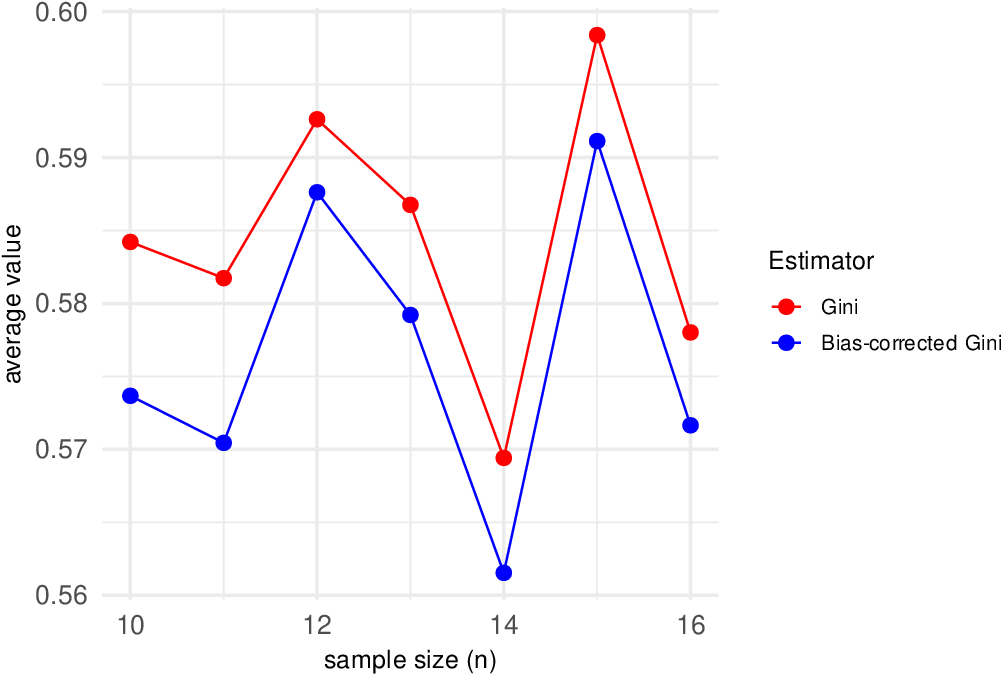}
    \caption{Comparison of standard Gini coefficient ($\widehat{G}$, red line) and bias-corrected Gini ($\widehat{G}_{\text{bc}}$, blue line) for different sample sizes $n$ ($\alpha_1=0.5$ and $\alpha_2= 2.0$).}
    \label{fig:gini_comparison1}
\end{figure}

\begin{figure}[!ht]
    \centering
    \includegraphics[scale=0.7]{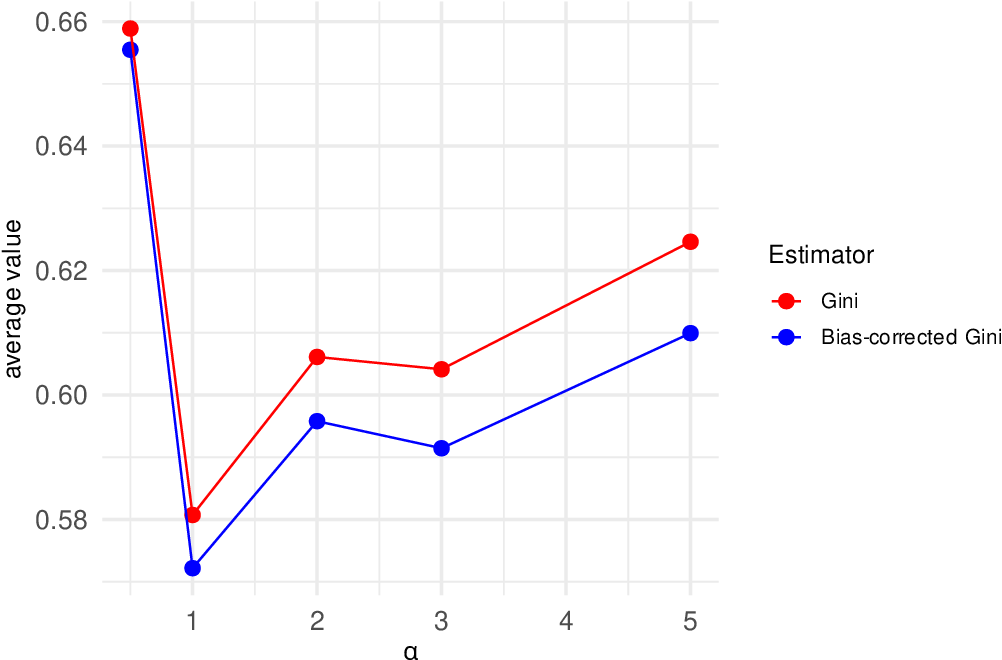}
    \caption{Comparison of standard Gini coefficient ($\widehat{G}$, red line) and bias-corrected Gini ($\widehat{G}_{\text{bc}}$, blue line) for different values of $\alpha_2$ ($\alpha_1=0.5$ and $n=15$).}
    \label{fig:gini_comparison2}
\end{figure}

\clearpage
		
\section{Concluding remarks}\label{sec:06}

In this paper, we have proposed the bias in the estimation of the Gini coefficient for gamma mixture populations. While the sample Gini coefficient is unbiased for single gamma distributions, we have demonstrated both theoretically and empirically that this property does not hold for gamma mixture models. We have derived a bias expression for the Gini coefficient for gamma mixture populations, which allowed us to propose a bias-corrected Gini estimator. An illustrative Monte Carlo simulation study has been carried out to evaluate the behavior of the bias-corrected Gini estimator. The results emphasized the that the standard Gini coefficient estimator exhibits a consistent upward bias when applied to gamma mixture populations, and the proposed bias correction effectively mitigates this issue, providing a more reliable estimator for gamma mixture distributions. As part of future research, it will be of interest to extend the study to multivariate Gini coefficients. Furthermore, alternative correction methods or the analysis to other mixture models can be developed; see \cite{Perez1986}. Work on these problems is currently in progress and we hope to report these findings in future.



	\paragraph*{Acknowledgements}
The research was supported in part by CNPq and CAPES grants from the Brazilian government.
	
	\paragraph*{Disclosure statement}
	There are no conflicts of interest to disclose.



\end{document}